\documentclass{IEEEtran}

\usepackage{amsmath,amssymb}
\usepackage{mathtools}
\usepackage{hyperref}

\usepackage{tikz}
\usetikzlibrary{matrix}

\def\F{{\mathbb F}}
\def\Z{{\mathbb Z}}

\def\diag{{\rm diag}}
\def\onemat{{\mathbf 1}}
\def\zeromat{{\mathbf 0}}
\def\GL{{\mathrm GL(n,\F_2)}}
\def\Sp{{\mathrm Sp}(2n,\F_2)}

\global\long\def\H{\textsc{H}}
\global\long\def\P{\textsc{P}}
\global\long\def\Z{\textsc{Z}}
\global\long\def\CNOT{\textsc{CNOT}}

\global\long\def\CZ{\textsc{CZ}}

\long\def\dummy#1{}

\newtheorem{theorem}{Theorem}
\newtheorem{definition}[theorem]{Definition}
\newtheorem{lemma}[theorem]{Lemma}
\newtheorem{corollary}[theorem]{Corollary}

\usepackage{xy}
\xyoption{matrix}
\xyoption{frame}
\xyoption{arrow}
\xyoption{arc}

\usepackage{ifpdf}
\ifpdf
\else
\PackageWarningNoLine{Qcircuit}{Qcircuit is loading in Postscript mode.  The Xy-pic options ps and dvips will be loaded.  If you wish to use other Postscript drivers for Xy-pic, you must modify the code in Qcircuit.tex}
\xyoption{ps}
\xyoption{dvips}
\fi

\entrymodifiers={!C\entrybox}

\newcommand{\ket}[1]{{\left\vert{#1}\right\rangle}}
\newcommand{\qw}[1][-1]{\ar @{-} [0,#1]}
\newcommand{\qwx}[1][-1]{\ar @{-} [#1,0]}

\newcommand{\gate}[1]{*+<.6em>{#1} \POS ="i","i"+UR;"i"+UL **\dir{-};"i"+DL **\dir{-};"i"+DR **\dir{-};"i"+UR **\dir{-},"i" \qw}

\newcommand{\control}{*!<0em,.025em>-=-<.2em>{\bullet}}

\newcommand{\ctrl}[1]{\control \qwx[#1] \qw}

\newcommand{\targ}{*+<.02em,.02em>{\xy ="i","i"-<.39em,0em>;"i"+<.39em,0em> **\dir{-}, "i"-<0em,.39em>;"i"+<0em,.39em> **\dir{-},"i"*\xycircle<.4em>{} \endxy} \qw}
\newcommand{\lstick}[1]{*!R!<.5em,0em>=<0em>{#1}}

\newcommand{\Qcircuit}{\xymatrix @*=<0em>}

\title{Shorter stabilizer circuits via Bruhat decomposition and quantum circuit transformations}

\author{%
Dmitri Maslov and Martin Roetteler
  \thanks{D.~Maslov is with the National Science Foundation, Alexandria, VA 22314, USA. M.~Roetteler is with Microsoft Research, Redmond, WA 98052, USA.}
\thanks{This material was based on work supported by the National Science Foundation, while DM working at the Foundation. Any opinion, finding, and conclusions or recommendations expressed in this material are those of the author and do not necessarily reflect the views of the National Science Foundation.}
}

\begin{document}

\maketitle 

\begin{abstract}
In this paper we improve the layered implementation of arbitrary stabilizer circuits introduced by Aaronson and Gottesman in Phys. Rev. A 70(052328), 2004: to obtain a general stabilizer circuit, we reduce their 11-stage computation -H-C-P-C-P-C-H-P-C-P-C- over the gate set consisting of Hadamard, Controlled-NOT, and Phase gates, into a $7$-stage computation of the form -C-CZ-P-H-P-CZ-C-.  We show arguments in support of using -CZ- stages over the -C- stages: not only the use of -CZ- stages allows a shorter layered expression, but -CZ- stages are simpler and appear to be easier to implement compared to the -C- stages.  Based on this decomposition, we develop a two-qubit gate depth-$(14n{-}4)$ implementation of stabilizer circuits over the gate library $\{\H,\P,\CNOT\}$, executable in the Linear Nearest Neighbor (LNN) architecture, improving best previously known depth-$25n$ circuit, also executable in the LNN architecture.  Our constructions rely on Bruhat decomposition of the symplectic group and on folding arbitrarily long sequences of the form $($-P-C-$)^m$ into a 3-stage computation -P-CZ-C-.  Our results include the reduction of the $11$-stage decomposition -H-C-P-C-P-C-H-P-C-P-C- into a $9$-stage decomposition of the form -C-P-C-P-H-C-P-C-P-.  This reduction is based on the Bruhat decomposition of the symplectic group.  This result also implies a new normal form for stabilizer circuits.  We show that a circuit in this normal form is optimal in the number of Hadamard gates used.  We also show that the normal form has an asymptotically optimal number of parameters.

\end{abstract}


\maketitle

\section{Introduction}
Stabilizer circuits are of particular interest in quantum information processing (QIP) due to their prominent role in fault tolerance \cite{ar:crss, www:g, bk:nc, ar:s}, the study of entanglement \cite{ar:bdsw, bk:nc}, and in evaluating quantum information processing proposals via randomized benchmarking \cite{ar:klr}, to name a few.  

Stabilizer circuits are composed of the Hadamard gate $\H$, Phase gate $\P$, and the controlled-NOT gate $\CNOT$ defined as
\[
\H:=\text{\footnotesize $\frac{1}{\sqrt{2}}$}\left[\begin{smallmatrix}
1 & 1\\
1 & -1
\end{smallmatrix}\right],\; \P:=\left[\begin{smallmatrix} 1 & 0 \\ 0 & i \end{smallmatrix}\right],\; \text{and } \CNOT:=\left[\begin{smallmatrix} 1 & 0 & 0 & 0 \\ 0 & 1 & 0 & 0 \\ 0 & 0 & 0 & 1 \\ 0 & 0 & 1 & 0 \end{smallmatrix}\right],
\]
where in case of $\H$ and $\P$ each of these gates is allowed to act on any of a given number $n$ of qubits, and on any pair of qubits in case of the $\CNOT$ gate. 

The stabilizer circuits over $n$ qubits, such as defined above form a finite group which is known to be equivalent \cite{ar:crss,ar:crss2} to the group of binary $2n \times 2n$ symplectic matrices, $\Sp$.  Knowing this equivalence allows to evaluate the stabilizer group size, through employing the well-known formula to calculate the number of elements in the respective symplectic group, $$|\Sp|={2^{n^2}\prod\limits_{j=1}^{n}(2^{2j}-1)}=2^{2n^2+O(n)}.$$ 

In this paper, we rely on the phase polynomial representation of $\{\P,\CNOT\}$ circuits.  Specifically, arbitrary quantum circuits over \textsc{P} and \textsc{CNOT} gates can be described in an alternate form, which we refer to as {\em phase polynomial description}, and vice versa, each phase polynomial description can be written as a \textsc{P} and \textsc{CNOT} gate circuit.  We use this result to induce circuit transformations via rewriting the respective phase polynomials.  We adopt the phase polynomial expression result from \cite{ar:ammr} to this paper as follows:

\begin{theorem}\label{thm:0}
Any circuit $C$ on $n$ qubits over $\{\P,\CNOT\}$ library with $k$ Phase gates can be described by the combination of a phase polynomial $p(x_1, x_2, ..., x_n)=f_1(x_1, x_2, ..., x_n) + f_2(x_1, x_2, ..., x_n) + \cdots + f_k(x_1, x_2, ..., x_n)$ and a linear reversible function $g(x_1, x_2, ..., x_n)$, such that the action of $C$ can be constructed as 
\[
C|x_1x_2...x_n\rangle = i^{p(x_1, x_2, ..., x_n)}|g(x_1, x_2, ..., x_n)\rangle, 
\]
where $i$ denotes the complex imaginary unit. Functions $f_j$ corresponding to the $j^{\text{th}}$ Phase gate are obtained from the circuit $C$ via devising Boolean linear functions computed by the \textsc{CNOT} gates in the circuit $C$ leading to the position of the respective Phase gate. 
\end{theorem}

In the following we focus on finding a short layered sequence of gates capable of representing an arbitrary stabilizer circuit over $n$ primary inputs.  The layers are defined as follows: 
\begin{itemize}
\item -H- layer contains all unitaries representable by arbitrary circuits composed of the Hadamard gates. Since $\H^2=Id$, and Hadamard gate is a single-qubit gate, -H- layer has zero or one gates acting on each of the respective qubits.  The number of distinct layers -H- on $n$ qubits is thus $2^n$.  We say -H- has $n$ Boolean degrees of freedom.
\item -P- layer is composed of an arbitrary set of Phase gates.  Since $\P^4=Id$, and the Phase gate is also a single-qubit gate, -P- layer has anywhere between zero to three gates on each of the respective qubits.  Note that $\P^2=\Z$, and therefore the gate sequence $\P\P$ may be better implemented as the Pauli-$\Z$ gate; $\P^3=\P^\dagger$, and frequently $\P^\dagger$ is constructible with the same cost as $\P$.  This means that the -P- layer is essentially analogous to the -H- layer in the sense that it consists of at most $n$ individual single-qubit gates.  The number of different unitaries represented by -P- layers on $n$ qubits is $2^{2n}$.  We say -P- has $2n$ Boolean degrees of freedom.  
\item -C- layer contains all unitaries computable by the $\CNOT$ gates.  The number of different -C- layers corresponds to the number of affine linear reversible functions, and it is equal to $\prod\limits_{j=0}^{n-1}(2^n-2^j)=2^{n^2+O(n)}$ \cite{ar:pmh}.  We say -C- has $n^2+O(n)$ Boolean degrees of freedom.
\item -CZ- layer contains all unitaries computable by the $\CZ$ gates, where $\CZ$ gate is defined as
\[\CZ:=\left[\begin{smallmatrix} 1 & 0 & 0 & 0 \\ 0 & 1 & 0 & 0 \\ 0 & 0 & 1 & 0 \\ 0 & 0 & 0 & -1 \end{smallmatrix}\right]
.\]  Since all $\CZ$ gates commute, and due to $\CZ$ being self-inverse, {\em i.e.}, $\CZ^2=Id$, the number of different unitaries computable by -CZ- layers is $\prod\limits_{j=1}^{n}2^{n-j}=2^{\frac{n^2}{2}+O(n)}$.  We say -CZ- has $\frac{n^2}{2}+O(n)$ Boolean degrees of freedom.
\end{itemize} 

Observe that the above count of the degrees of freedom suggests that -P- and -H- layers are ``simple''.  Indeed, each requires no more than the linear number of single-qubit gates to be constructed via a circuit.  The number of the degrees of freedom in -C- and -CZ- stages is quadratic in $n$.  Other than the two-qubit gates often being more expensive than the single-qubit gates \cite{ar:deb,www:IBM}, the comparison of the degrees of freedom suggests that we will need more of the respective gates to construct each such stage.  The -CZ- layer has roughly half the number of the degrees of freedom compared to the -C- layer.  We may thus reasonably expect that the -CZ- layer can be easier to obtain.  

Unlike the -C- circuits, the problem of optimizing -CZ- circuits does not seem to have been studied in the literature.  Part of the reason could be due to the $\CZ$ gate complexity of -CZ- circuits being a very inconspicuous problem to study: indeed, worst case optimal circuit has $\frac{(n-1)n}{2}$ $\CZ$ gates, and optimal circuits are easy to construct, as they are determined by the presence or lack of $\CZ$ gates acting on the individual pairs of qubits.  However, we claim that using only $\CZ$ gates to construct -CZ- layer is not the best solution, and a better approach would be to also employ the $\CNOT$ and $\P$ gates.  Indeed, both $\CNOT$ and $\CZ$ gates must have a comparable cost of the implementation, since they are related by the formula $\CNOT(a,b)=\H(b)\CZ(a,b)\H(b)$, and single-qubit gates are ``easy'' \cite{ar:deb, www:IBM}.  $\CZ$ is furthermore the elementary gate in superconducting circuits QIP \cite{ar:ggz}, and as such, technically, it costs less than the $\CNOT$, and in the trapped ions QIP the costs of the two are comparable \cite{ar:m}.  Further discussion of the relation of implementation costs between -C- and -CZ- layers is postponed to Section \ref{sec:cvscz}.

The different layers can be interleaved to obtain stabilizer circuits not computable by a single layer.  A remarkable result of \cite{ar:ag} shows that 11 stages over a computation of the form -H-C-P-C-P-C-H-P-C-P-C- suffices to compute an arbitrary stabilizer circuit.  The number of Boolean degrees of freedom in the group of stabilizer unitaries, defined as the logarithm base-2 of their total count, is given by the formula $\log_2{|\Sp|}=2n^2+O(n)$.  This suggests that the 11-stage circuit by Aaronson and Gottesman \cite{ar:ag} is suboptimal, as it relies on $5n^2+O(n)$ degrees of freedom, whereas only $2n^2+O(n)$ are necessary.  Indeed, we find (Section \ref{sec:6}) a shorter 9-stage decomposition of the form -P-C-P-C-H-C-P-C-P- in which all -C- stages correspond to upper triangular matrices having $n^2/2$ degrees of freedom each, leading to an asymptotically tight parameterization of all stabilizer circuits. 

{\it Notation.} We denote with $\GL$ the group of invertible $n\times n$ matrices, with $S_n$ the full permutation group on $n$ letters, and with $\diag(A,B)$ the (block) diagonal operator that has diagonal elements $A$ and $B$.

\section{$($-P-C-$)^m$ circuits}

In this section we show that an arbitrary length $n$-qubit computation described by the stages -P-C-P-C-...-P-C- folds into an equivalent three-stage computation -P-CZ-C-.

\begin{theorem}\label{thm:1}
$(\text{-P-C-})^m=\text{-P-CZ-C-}$.
\end{theorem}
\begin{proof}
A $($-P-C-$)^m$ circuit has no more than $k\leq 3nm$ Phase gates.  Name those gates $\P_{j=1..k}$, denote Boolean linear functions they apply phases to as $f_{j=1..k}(x_1,x_2,...,x_n)$, and name the reversible linear function computed by \linebreak $($-P-C-$)^m$ (Theorem \ref{thm:0}) as $g(x_1,x_2,...,x_n)$. Phase polynomial computed by the original circuit is $f_1(x_1,x_2,...,x_n)+f_2(x_1,x_2,...,x_n)+...+f_k(x_1,x_2,...,x_n)$.  We will next transform phase polynomial to an equivalent one, that will be easier to write as a compact circuit.  To accomplish this, observe that $i^{a+b+c+(a\oplus b) + (a\oplus c) + (b\oplus c) + (a\oplus b \oplus c)} = i^4= 1$, where $a$, $b$, and $c$ are arbitrary Boolean linear functions of the primary variables.  This equality can be verified by inspection through trying all 8 possible combinations for Boolean values $a$, $b$, and $c$. The equality can be rewritten as 
\begin{equation}\label{phasetmpl}
i^{a\oplus b \oplus c} = i^{3a+3b+3c+3(a\oplus b)+3(a\oplus c)+3(b\oplus c)},
\end{equation} 
suggesting how it will be used.  The following algorithm takes $n{-}2$ steps.

Step $n$. Consolidate terms in the phase polynomial $f_1(x_1,x_2,...,x_n)+f_2(x_1,x_2,...,x_n)+...+f_k(x_1,x_2,...,x_n)$ by replacing $uf_j(x_1,x_2,...,x_n)+vf_k(x_1,x_2,...,x_n)$ with $(u+v \bmod 4)f_j(x_1,x_2,...,x_n)$ whenever $f_j=f_k$.  Once done, look for $f_j=x_1 \oplus x_2 \oplus ... \oplus x_n$, being the maximal length linear function of the primary inputs.  If no such function found, move to the next step.  If it is found with a non-zero coefficient $u$, as an additive term $u(x_1 \oplus x_2 \oplus ... \oplus x_n)$, replace it by the equivalent 6-term mixed arithmetic polynomial $(4-u)x_1 + (4-u)x_2 + (4-u)(x_3 \oplus x_4 \oplus ... \oplus x_n) + (4-u)(x_1 \oplus x_2) + (4-u)(x_1 \oplus x_3 \oplus x_4 \oplus ... \oplus x_n) + (4-u)(x_2 \oplus x_3 \oplus ... \oplus x_n)$. This transformation is derived from eq.~(\ref{phasetmpl}) by assigning $a=x_1$, $b=x_2$, and $c=x_3 \oplus x_4 \oplus... \oplus x_n$.  Consolidate all equal terms.  The transformed phase polynomial is equivalent to the original one in the sense of the overall combination of phases it prescribes to compute, however, it is expressed over linear terms with at most $n{-}1$ variables.

Step $s$, $s=(n-1)..3$. From the previous step we have phase polynomial of the form $u^\prime_1f^\prime_1(x_1,x_2,...,x_n)+u^\prime_2f^\prime_2(x_1,x_2,...,$ $x_n)+...+u^\prime_{k^\prime}f^\prime_{k^\prime}(x_1,x_2,...,x_n).$ By construction it is guaranteed that the functions $f^\prime_{j=1..k^\prime}$ EXOR no more than $s$ literals.  For each $f^\prime_j=x_{j_1} \oplus x_{j_2} \oplus ... \oplus x_{j_s}$, with the coefficient $u^{\prime}_j \not\equiv 0 \bmod 4$ replace this term with the sum of six terms, each having no more than $s{-}1$ literals by using eq.~(\ref{phasetmpl}) and setting $a$, $b$, and $c$ to carry linear functions over the non-overlapping non-empty subsets of $\{x_{j_1}, x_{j_2}, ..., x_{j_s}\}$ whose union gives the entire set $\{x_{j_1}, x_{j_2}, ..., x_{j_s}\}$. Value $s{=}3$ marks the last opportunity to break down a term in the phase polynomial expression into a set of terms over smaller numbers of variables.  Upon completion of this step, the linear functions participating in the phase polynomial expression contain at most two literals each.

The transformed phase polynomial description of the original circuit now has the following form: 
phase polynomial $\sum\limits_{j=1}^{n}u_jx_j + \sum\limits_{j=1}^{n}\sum\limits_{k=j+1}^{n}u_{j,k}(x_j \oplus x_k)$, where $u_{\cdot}, u_{\cdot,\cdot} \in \mathbb{Z}_4$, and the linear reversible function $g(x_1,x_2,...,x_n)$. We next show how to implement such a unitary as a -P-CZ-C- circuit, focusing separately on the phase polynomial and the linear reversible part. We synthesize individual terms in the phase polynomial as follows. 
\begin{itemize}
\item For $j=1..n,$ the term $u_jx_j$ is obtained as the single-qubit gate circuit $\P^{u_j}(x_j)$;
\item For $j=1..n$, $k=j{+}1 ..n$, the term $u_{j,k}(x_j \oplus x_k)$ is obtained as follows: 
\begin{itemize}
\item if $u_{j,k}\equiv 2 \bmod 4$, by the circuit $\P^2(x_j)\P^2(x_k) = \Z(x_j)\Z(x_k)$;
\item if $u_{j,k}\equiv 1 \text{ or } 3 \bmod 4$, by the circuit  $\P^{u_{j,k}}(x_j)\P^{u_{j,k}}(x_k) \CZ(x_j,x_k)$.
\end{itemize} 
\end{itemize}
The resulting circuit contains $\P$ and $\CZ$ gates; it implements phase polynomial $\sum\limits_{j=1}^{n}u_jx_j + \sum\limits_{j=1}^{n}\sum\limits_{k=j+1}^{n}u_{j,k}(x_j \oplus x_k)$ and the identity linear reversible function. Since all $\P$ and $\CZ$ gates commute, Phase gates can be collected on the left side of the circuit. This results in the ability to express phase polynomial construction as a -P-CZ- circuit.  We conclude the entire construction via obtaining the linear reversible function $g(x_1,x_2,...,x_n)$ as a -C- stage, with the overall computation described as a -P-CZ-C- circuit.
\end{proof}

Note that -P-CZ-C- can also be written as -C-P-CZ-, if one first synthesizes the linear reversible function $g(x_1,x_2,...,x_n)=(g_1(x_1,x_2,...,x_n),g_2(x_1,x_2,...,x_n),...,g_n(x_1,x_2,...,x_n))$, and then expresses the phase polynomial in terms of the degree-2 terms over the set $\{g_1,g_2,...,g_n\}$. Other ways to write such a computation include -CZ-P-C- and -C-CZ-P-, that are obtained from the first two by commuting -P- and -CZ- stages.

\begin{corollary}
-H-C-P-C-P-C-H-P-C-P-C- \cite{ar:ag} = -H-C-CZ-P-H-P-CZ-C-.
\end{corollary}

\section{-C- vs -CZ-}\label{sec:cvscz}

We have previously noted that $\CNOT$ and $\CZ$ gates have a comparable cost as far as their implementation within some QIP proposals is concerned.  In this section, we study  $\{\P, \CZ, \CNOT\}$ implementations of stages -C- and {-CZ-}.  The goal is to provide further evidence in support of the statement that -CZ- can be thought of as a simpler stage compared to the -C- stage, and going beyond counting the degrees of freedom argument.

\begin{lemma}
Optimal quantum circuit over $\{\CZ\}$ library for a -CZ- stage has at most $\frac{n(n-1)}{2}$ $\CZ$ gates.
\end{lemma}

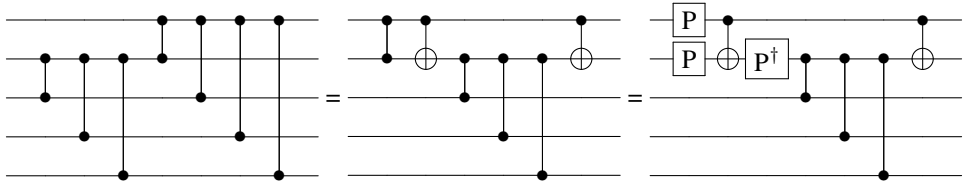
\begin{figure*}[t]
\begin{equation*}
\Qcircuit @C=1.23em @R=1.23em @! {
& \qw 		& \qw 		& \qw 		& \ctrl{1}	& \ctrl{2}  & \ctrl{3}	& \ctrl{4}	& \qw \\
& \ctrl{1}  & \ctrl{2}	& \ctrl{3}	& \ctrl{-1}	& \qw 		& \qw 		& \qw 		& \qw \\
& \ctrl{-1} & \qw 		& \qw 		& \qw 		& \ctrl{-2} & \qw 		& \qw 		& \qw \\
& \qw  		& \ctrl{-2}	& \qw		& \qw 		& \qw  		& \ctrl{-3}	& \qw		& \qw \\
& \qw 		& \qw 		& \ctrl{-3}	& \qw 		& \qw 		& \qw 		& \ctrl{-4}	& \qw 
} 
\raisebox{-2.95em}{\hspace{1mm}=\hspace{1mm}}
\Qcircuit @C=0.65em @R=0.65em @! {
& \ctrl{1}	& \ctrl{1} 	& \qw 		& \qw 		& \qw 		& \ctrl{1} 	& \qw \\
& \ctrl{-1}	& \targ		& \ctrl{1}  & \ctrl{2}	& \ctrl{3}	& \targ		& \qw \\
& \qw 		& \qw		& \ctrl{-1} & \qw 		& \qw 		& \qw		& \qw \\
& \qw 		& \qw		& \qw  		& \ctrl{-2}	& \qw		& \qw		& \qw \\
& \qw 		& \qw		& \qw 		& \qw 		& \ctrl{-3}	& \qw		& \qw 
} 
\raisebox{-2.95em}{\hspace{1mm}=\hspace{1mm}}
\Qcircuit @C=-0.15em @R=-0.15em @! {
& \gate{\P}	& \ctrl{1} 	& \qw 				& \qw 		& \qw 		& \qw 		& \ctrl{1} 	& \qw \\
& \gate{\P}	& \targ		& \gate{\P^\dagger}	& \ctrl{1}  & \ctrl{2}	& \ctrl{3}	& \targ		& \qw \\
& \qw 		& \qw		& \qw 				& \ctrl{-1} & \qw 		& \qw 		& \qw		& \qw \\
& \qw 		& \qw		& \qw 				& \qw  		& \ctrl{-2}	& \qw		& \qw		& \qw \\
& \qw 		& \qw		& \qw 				& \qw 		& \qw 		& \ctrl{-3}	& \qw		& \qw 
} 
\end{equation*}
\caption{\label{fig:exampleIdentities} Circuit identities illustrating rewriting of the -CZ- circuits.}
\end{figure*}

Indeed, all $\CZ$ gates commute, which limits the expressive power of the circuits over $\CZ$ gates.  However, once we add the non-commuting $\CNOT$ gate, and after that the Phase gate, the situation changes.  We can now implement -CZ- circuits more efficiently, such as illustrated by the circuit identities shown in Fig.~\ref{fig:exampleIdentities}.  The unitary implemented by the circuitry shown in Fig.~\ref{fig:exampleIdentities} requires $7$ $\CZ$ gates as a $\{\CZ\}$ circuit, $6$ gates as a $\{\CZ,\CNOT\}$ circuit, and only $5$ two-qubit gates as a $\{\P,\CZ,\CNOT\}$ circuit.  This illustrates that the $\CNOT$ and $\P$ gates are important in constructing efficient -CZ- circuits.

We may consider adding the $\P$ and $\CZ$ gates to the $\{\CNOT\}$ library in hopes of constructing more efficient circuits implementing the -C- stage.  However, as the following lemma shows, this does not help.
\begin{lemma}
Any $\{\P, \CZ, \CNOT\}$ circuit implementing an element of the layer -C- using a non-zero number of $\P$ and $\CZ$ gates is suboptimal.
\end{lemma}
\begin{proof}
Each $\P$ gate applied to a qubit $x$ can be expressed as a phase polynomial $1\cdot x$ over the identity reversible linear function.  Each $\CZ$ gate applied to a set of qubits $y$ and $z$ can be expressed as a phase polynomial $y + z + 3(y \oplus z)$ and the identity reversible function.  Removing all $\P$ and $\CZ$ gates from the given circuit thus modifies only the phase polynomial part of its phase polynomial description.  Removing all $\P$ and $\CZ$ gates from the $\{\P, \CZ, \CNOT\}$ circuit guarantees that the phase polynomial of the resulting circuit equals to the identity, such as required in the -C- stage.  This results in the construction of a shorter circuit in cases when the original $\P$ and $\CZ$ gate count was non-zero.
\end{proof}

We next show in Table \ref{tab:main} optimal counts and upper bounds on the number of gates it takes to synthesize the most difficult function from stages -C- and -CZ- for some small $n$. Observe how the two-qubit gate counts for the -CZ- stage, when constructed as a circuit over  $\{\P, \CZ, \CNOT\}$ library, remain lower than those for the -C- stage. 

\begin{table}[ht]
\centering 
\begin{tabular}{c|cc|cc} 
  & \multicolumn{2}{c|}{-CZ-} 		& \multicolumn{2}{c}{-C-} \\
n & $\{\CZ\}$ 	& $\{\P, \CZ,\CNOT\}$ 	& $\{\CNOT\}$ 	& $\{\P, \CZ,\CNOT\}$ \\ \hline
2 & 1			& 1					& 3				& 3 \\
3 & 3			& 3					& 6 			& 6  \\
4 & 6			& 5					& 9 			& 9  \\
5 & 10			& 7					& 12			& 12  \\
\end{tabular} 
\caption{Gate counts required to implement arbitrary -CZ- and -C- stages for some small $n$: optimal -CZ- stage gate counts as circuits over $\{\CZ\}$, upper bounds on the two-qubit gate count for -CZ- over $\{\P, \CZ,\CNOT\}$, achieved based on the application of identities from Fig.~(\ref{fig:exampleIdentities}) applied to circuits with $\CZ$ gates, and optimal $\{\CNOT\}$ and $\{\P, \CZ,\CNOT\}$ two-qubit gate counts for stage -C-.} \label{tab:main}
\end{table}

In \cite{ar:pmh} an asymptotically optimal algorithm  for $\{\CNOT\}$ synthesis of arbitrary -C- stage functions was reported, that leads to the worst case gate complexity of $O\left(\frac{n^2}{\log n}\right)$.  It is possible that an asymptotically optimal algorithm for $\{\P, \CZ, \CNOT\}$ circuits implementing arbitrary -CZ- stage functions can be developed, at which point its complexity has to be $O\left(\frac{n^2}{\log n}\right)$.  To determine which of the two results in shorter circuits, one has to develop constants in front of the leading complexity terms. 

\bigskip
We point out that gate count is only one of several possible metrics of efficiency. For instance, two-qubit gate depth over Linear Nearest Neighbour (LNN) architecture is also an important metric of efficiency.  This metric has been applied in \cite{ar:kms} to show an asymptotically optimal upper bound of $5n$ $\CNOT$ layers required to obtain an arbitrary -C- stage.  

Define $\widehat{\text{-CZ-}}$ to be -CZ- accompanied by the complete qubit reversal ({\em i.e.}, the linear reversible mapping $\ket{x_1x_2...x_n} \mapsto \ket{x_nx_{n-1}...x_1}$).  We next show that $\widehat{\text{-CZ-}}$ can be executed as a two-qubit gate depth-$(2n+2)$ computation over LNN.  This result will be used to reduce depth in the implementation of arbitrary stabilizer circuits.

\begin{theorem}\label{thm:2}
$\widehat{\text{-CZ-}}$ can be implemented as a $\CNOT$ depth-$(2n+2)$ circuit.
\end{theorem}
\begin{proof}
Consider phase polynomial description of the circuit $\widehat{\text{-CZ-}}$. However, rather than describe both parts of the expression, phase polynomial itself and the linear reversible transformation, over the set of primary variables, we will describe phase polynomial over the variables $y_1,y_2,...,y_n$ defined as follows:
\begin{eqnarray*}
y_1:=x_1, \\
y_2:=x_1 \oplus x_2, \\
..., \\
y_n:= x_1 \oplus x_2 \oplus ... \oplus x_n.
\end{eqnarray*}
This constitutes the change of basis $\{x_1,x_2,...,x_n\} \mapsto \{y_1,y_2,...,y_n\}$.  Similarly to how it was done in the proof of Theorem \ref{thm:1}, we reduce phase polynomial representation of $\widehat{\text{-CZ-}}$ to the application of Phase gates to the EXORs of pairs and the individual variables from the set $\{y_1,y_2,...,y_n\}$, 
\begin{equation}\label{eq:alljk}
\sum\limits_{j=1}^{n}u_jy_j + \sum\limits_{j=1}^{n}\sum\limits_{k=j+1}^{n}u_{j,k}(y_j \oplus y_k),
\end{equation}
and the linear reversible function $g(x_1,x_2,...,x_n): \ket{x_1x_2...x_n} \mapsto \ket{x_nx_{n-1}...x_1}$.  Observe that $y_j \oplus y_k = x_j \oplus x_{j+1} \oplus ... \oplus x_k$, and thereby this linear function can be encoded by the integer segment $[j,k]$. The primary variable $x_j$ admits the encoding $[j,j]$.  We use this notation next.  In the following we implement the pair of the phase polynomial expression and the reversal of qubits (a linear reversible function) via a quantum circuit.

\begin{figure*}[t]
\centerline{
\Qcircuit @C=0.97em @R=0.97em @! {
&&&&& [6,\;\;] &&&&&& [\;\;, \;\;]	&&&&&& [\;\;,3] &&&&&&& \\
&&&&& [4,\;\;] &&&&&& [\;\;, \;\;]	&&&&&& [\;\;,3] &&&&&&& \\
&&&&& [4,\;\;] &&&&&& [6, 3]		&&&&&& [\;\;,5] &&&&&&& \\
&&&&& [2,\;\;] &&&&&& [4, 3]		&&&&&& [\;\;,5] &&&&&&& \\
\lstick{[1,1]}	& \ctrl{1}	& \qw 		& \targ		& \qw		& [2,3] &
				& \ctrl{1}	& \qw 		& \targ		& \qw		& [4,5] &
				& \ctrl{1}	& \qw 		& \targ		& \qw		& [6,7] & 
				& \ctrl{1}	& \qw 		& \targ		& \qw		& \qw & {[7,7]} \\
\lstick{[2,2]}	& \targ 	& \targ		& \ctrl{-1}	& \ctrl{1}	& [1,3] &
				& \targ 	& \targ		& \ctrl{-1}	& \ctrl{1}	& [2,5] &
				& \targ 	& \targ		& \ctrl{-1}	& \ctrl{1}	& [4,7] &
				& \targ 	& \targ		& \ctrl{-1}	& \ctrl{1}	& \qw & {[6,6]} \\
\lstick{[3,3]}	& \ctrl{1}	& \ctrl{-1}	& \targ		& \targ		& [1,5] &
				& \ctrl{1}	& \ctrl{-1}	& \targ		& \targ		& [2,7] &
				& \ctrl{1}	& \ctrl{-1}	& \targ		& \targ		& [4,6] & 
				& \ctrl{1}	& \ctrl{-1}	& \targ		& \targ		& \qw & {[5,5]} \\
\lstick{[4,4]}	& \targ 	& \targ		& \ctrl{-1}	& \ctrl{1}	& [3,5] &
				& \targ 	& \targ		& \ctrl{-1}	& \ctrl{1}	& [1,7] &
				& \targ 	& \targ		& \ctrl{-1}	& \ctrl{1}	& [2,6] & 
				& \targ 	& \targ		& \ctrl{-1}	& \ctrl{1}	& \qw & {[4,4]} \\
\lstick{[5,5]}	& \ctrl{1}	& \ctrl{-1}	& \targ		& \targ		& [3,7] &
				& \ctrl{1}	& \ctrl{-1}	& \targ		& \targ		& [1,6] &
				& \ctrl{1}	& \ctrl{-1}	& \targ		& \targ		& [2,4] &
				& \ctrl{1}	& \ctrl{-1}	& \targ		& \targ		& \qw & {[3,3]} \\
\lstick{[6,6]}	& \targ 	& \targ		& \ctrl{-1}	& \ctrl{1}	& [5,7] &
				& \targ 	& \targ		& \ctrl{-1}	& \ctrl{1}	& [3,6] &
				& \targ 	& \targ		& \ctrl{-1}	& \ctrl{1}	& [1,4] &
				& \targ 	& \targ		& \ctrl{-1}	& \ctrl{1}	& \qw & {[2,2]} \\
\lstick{[7,7]}	& \qw		& \ctrl{-1}	& \qw		& \targ		& [5,6] &
				& \qw		& \ctrl{-1}	& \qw		& \targ		& [3,4] &
				& \qw		& \ctrl{-1}	& \qw		& \targ		& [1,2] & 
				& \qw		& \ctrl{-1}	& \qw		& \targ		& \qw & {[1,1]} \\
&&&&& [\;\;,6] &&&&&& [5, 4]		&&&&&& [3,\;\;] &&&&&&& \\
&&&&& [\;\;,4] &&&&&& [5, 2]		&&&&&& [3,\;\;] &&&&&&& \\
&&&&& [\;\;,4] &&&&&& [\;\;, \;\;]	&&&&&& [5,\;\;] &&&&&&& \\
&&&&& [\;\;,2] &&&&&& [\;\;, \;\;]	&&&&&& [5,\;\;] &&&&&&& \\
&&&&& \downarrow \;\; \uparrow &&&&&& \downarrow \;\; \uparrow	&&&&&& \downarrow \;\; \uparrow &&&&&&&
}
}
\caption{Constructing $\widehat{\text{-CZ-}}$ for $n=7$. The circuit uses $m+1=\frac{n-1}{2}+1=4$ depth-$4$ stages $S$. Patterns $Pj$ and $Pk$ are $(6,4,4,2,2,1,1,3,3,5,5)$ and $(3,3,5,5,7,7,6,6,4,4,2)$, correspondingly.  Arrows $\downarrow$ and $\uparrow$ show the direction of the 2-position shifts of the respective patterns. $[j,k]$ denotes linear function $x_j \oplus x_{j+1} \oplus ... \oplus x_k$ of primary variables, that accepts the application of Phase gates to, so long as contained to within the circuit.  A total of $4$ Phase gate stages is required; Phase gates can be applied to the individual literals selectively in the beginning or at the end of the circuit.}
\label{fig:ccz}
\end{figure*}
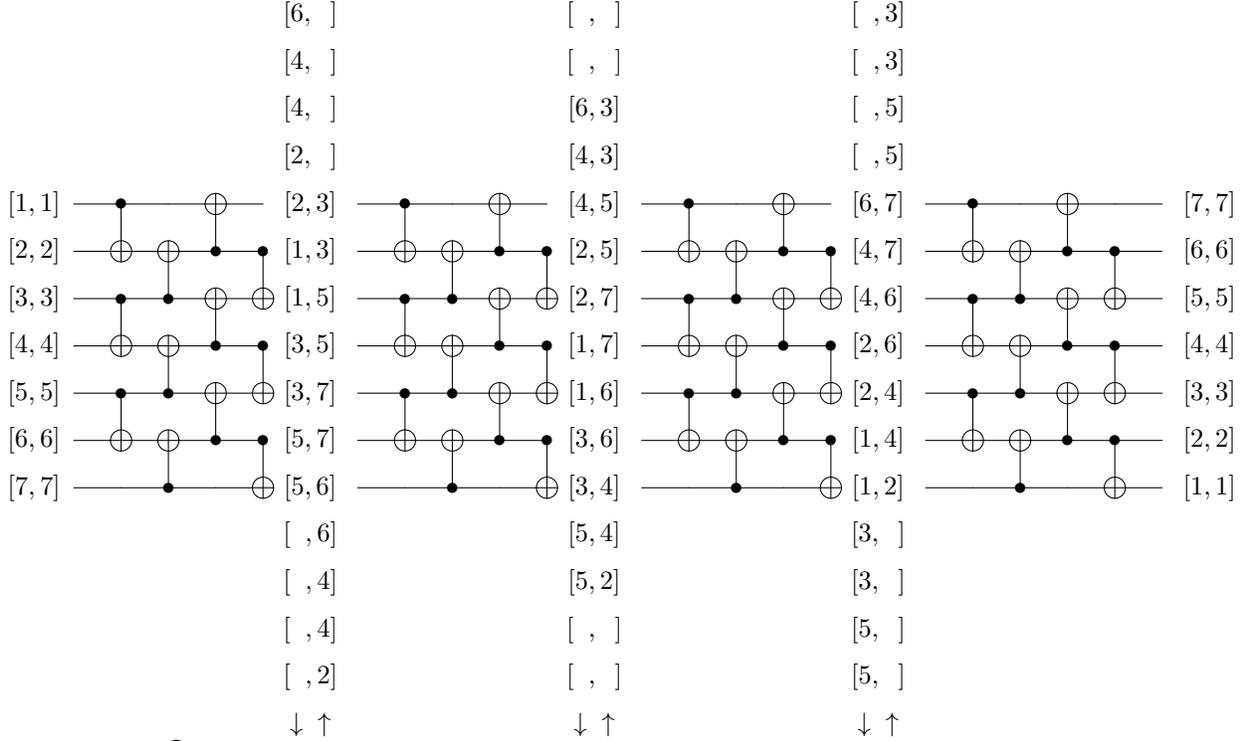

Observe that the swapping operation $g(x_1,x_2,...,x_n): \ket{x_1x_2...x_n} \mapsto \ket{x_nx_{n-1}...x_1}$ can be implemented as a circuit similar to the one from Theorem 5.1 \cite{ar:kms} in depth $2n{+}2$.  The rest of the proof concerns the ability to insert Phase gates in the circuit accomplishing the reversal of qubits such as to allow the implementation of each term in the phase polynomial, eq.~(\ref{eq:alljk}).

Since our qubit reversal circuit is slightly different from the one used in \cite{ar:kms}, and we explore its structure more extensively, we describe it next.  It consists of $n{+}1$ alternating stages, $S_1$ and $S_2$, where 

\begin{align*}
S_1:= & \; \CNOT(x_1;x_2)\CNOT(x_3;x_4)...\CNOT(x_{n-2};x_{n-1})  \\ 
     & \cdot \; \CNOT(x_3;x_2)\CNOT(x_5;x_4)...\CNOT(x_n;x_{n-1}) 
\end{align*}
for odd $n$, and 
\begin{align*}
S_1:= & \; \CNOT(x_1;x_2)\CNOT(x_3;x_4)...\CNOT(x_{n-1};x_n)  \\ 
     & \cdot \; \CNOT(x_3;x_2)\CNOT(x_5;x_4)...\CNOT(x_{n-1};x_{n-2})
\end{align*}
for even $n$, is a depth-$2$ circuit composed with the $\CNOT$ gates. Similarly, \begin{align*}
S_2:= & \; \CNOT(x_2;x_1)\CNOT(x_4;x_3)...\CNOT(x_{n-1};x_{n-2}) \\    
     & \cdot \; \CNOT(x_2;x_3)\CNOT(x_4;x_5)...\CNOT(x_{n-1};x_{n}) 
\end{align*} 
for odd $n$, and
\begin{align*}
S_2:= & \; \CNOT(x_2;x_1)\CNOT(x_4;x_3)...\CNOT(x_{n};x_{n-1}) \\
     & \cdot \; \CNOT(x_2;x_3)\CNOT(x_4;x_5)...\CNOT(x_{n-2};x_{n-1}) 
\end{align*}
for even $n$, is also a depth-$2$ circuit composed with the $\CNOT$ gates.  We refer to the concatenation of $S_1$ and $S_2$ as $S$.  The goal is to show that after $\lceil \frac{n}{2} \rceil$ applications of the circuit $S$ we are able to cycle through all $\frac{n(n+1)}{2}$ linear functions $[j,k]$, $j \leq k$. 

The remainder of the proof works slightly differently depending on the parity of $n$.  First, choose odd $n=2m{+}1$.  Consider two patterns of length $2n{-}3$,
\[
Pj:= (n{-}1,n{-}3,n{-}3,...,4,4,2,2,1,1,3,3,...,n{-}2,n{-}2)
\]
and 
\[ 
Pk:= (3,3,5,5,...,n,n,n{-}1,n{-}1,n{-}3,n{-}3,...6,6,4,4,2).
\]
Observe by inspection that the $i^{\text{th}}$ linear function computed by the single application of the stage $S$ is given by the formula $[Pj(n-3+i),Pk(i)]$, where $Pj(l)$ and $Pk(l)$ return $l^{\text{th}}$ component of the respective pattern.  It may further be observed, via direct multiplication by the linear reversible matrix corresponding to the transformation $S$, that the $i^{\text{th}}$ component upon $t$ ($t\leq m$) applications of the circuit $S$ is computable by the following formula, $[Pj(n-1-2t+i),Pk(2t-2+i)]=[Pj(n-3-2(t-1)+i),Pk(2(t-1)+i)]$. A simple visual explanation can be given: at each application of $S$ pattern $Pj$ is shifted by two positions to the left (down, Fig.~\ref{fig:ccz}), whereas pattern $Pk$ gets shifted by two positions to the right (up, Fig.~\ref{fig:ccz}).

Observe that every $[j,k]$, $j=1..n, k=1..n, j\leq k$ is being generated. Indeed, a given $[j,k]$ may only be generated at most once by the $0$ to $m$ applications of the circuit $S$. This is because once a given $j$ meets a given $k$ for the first time, at each following step, the respective value $k$ gets shifted away from $j$ to never meet again.  We next employ the counting argument to show that all functions $[j,k]$ are generated.  Indeed, the total number of functions generated by $0$ to $m$ applications of the stage $S$ is $(m+1)n=\left(\frac{n-1}{2}+1\right)n=\frac{n(n+1)}{2}$, each linear function generated is of the type $[j,k]$ ($j=1..n, k=1..n, j\leq k$), none of which can be generated more than once, and their total number is $\frac{n(n+1)}{2}$.  This means that every $[j,k]$ is generated. 

We illustrate the construction of the circuit implementing $\widehat{\text{-CZ-}}$ for $n=7$ in Fig.~\ref{fig:ccz}.

For even $n=2m$ the construction works similarly. The patterns $Pj$ and $Pk$ are $(n,n-2,n-2,n-4,n-4,...,2,2,1,1,3,3...,n-3,n-3,n-1)$ and $(3,3,5,5,...,n-1,n-1,n,n,n-2,n-2,...,4,4,2,2)$, respectively. The formula for computing the linear function $[j,k]$ for $i^{\text{th}}$ coordinate after $t$ applications of $S$ is $[Pj(n-2t+i),Pk(2t-2+i)]$. After $m$ applications of the circuit $S$ we generate linear functions $x_n, x_{n-1}, ..., x_4, x_2$ in addition to the $m$ new linear functions of the form $[j,k]$ ($j<k$). 
\end{proof}

To consider circuit depth makes most sense when applied to measure depth across most computationally intensive operations.  In both of the two leading approaches to quantum information processing, and limiting the attention to fully programmable universal quantum machines, superconducting circuits \cite{www:IBM} and trapped ions \cite{ar:deb}, the two-qubit gates take longer to execute and are associated with lower fidelity.  As such, they constitute the most expensive resource and motivate our choice to measure depth in terms of the two-qubit operations.  The selection of the LNN architecture to measure the depth over is motivated by the desire to restrict arbitrary interaction patterns to a reasonable set.  Both superconducting and trapped ions qubit-to-qubit connectivity patterns  \cite{ar:deb, www:IBM} are furthermore such that they allow embedding the linear chain in them.  

A further observation is that the two-qubit $\CNOT$ gate may not be native to a physical implementation, and therefore the $\CNOT$ implementation may likely use correcting single-qubit gates before and after using a specific two-qubit interaction.  This means that interleaving the two-qubit gates with the single-qubit gates such as done in the proof of Theorem \ref{thm:2} may not increase the depth, and restricting depth calculation to just the two-qubit stages is appropriate.  We did, however, report enough detail to develop depth figure over both single- and two-qubit gates for the implementations of stabilizer circuits relying on our construction. 

\begin{corollary}\label{col:2}
Arbitrary $n$-qubit stabilizer unitary can be executed in two-qubit gate depth $14n-4$ as an $\{\H, \P, \CNOT\}$ circuit over the LNN architecture.
\end{corollary}
\begin{proof}
Firstly, observe that -H-C-CZ-P-H-P-CZ-C $=$ -H-C$\widehat{\text{-CZ-}}$P-H-P$\widehat{\text{-CZ-}}$C-. This is because both $\widehat{\text{-CZ-}}$ stages reverse the order of qubits, and therefore the effect of the qubit reversal cancels out. The two-qubit gate depth of the -C- stage is $5n$ \cite{ar:kms}, and the two-qubit gate depth of the $\widehat{\text{-CZ-}}$ stage is $2n{+}2$, per Theorem \ref{thm:2}.  This means that the overall two-qubit gate depth is $14n{+}4$.  This number can be reduced somewhat by the following two observations.  Name individual stages in the target decomposition as follows, -H-C$_1$$\widehat{\text{-CZ$_1$-}}$P-H-P$\widehat{\text{-CZ$_2$-}}$C$_2$-. 
Using the construction in Theorem \ref{thm:2}, we can implement $\widehat{\text{-CZ$_1$-}}$ without the first $S$ circuit through applying Phase gates at the end of it (see Fig.~\ref{fig:ccz} for illustration).  The first $S$ circuit can then be combined with the -C$_1$- stage preceding it.  This results in the saving of $4$ layers of two-qubit computations.  Similarly, $\widehat{\text{-CZ$_2$-}}$ can be implemented up to $S$ if it is implemented in reverse, and phases are applied in the beginning (the end, but invert the circuit).  This allows to merge depth-$4$ computation $S$ with the stage -C$_2$- that follows.  These two modifications result in the improved depth figure of $14n{-}4$.  
\end{proof}

Observe how the aggregate contribution to the depth from both -CZ- stages used in our construction, ${\sim}4n$, is less than that from a single -C- stage, $5n$.  The result of \cite{ar:kms} can be applied to the 11-stage decomposition -H-C-P-C-P-C-H-P-C-P-C- of \cite{ar:ag} to obtain a two-qubit gate depth-$25n$ LNN-executable implementation of an arbitrary stabilizer unitary.  In comparison, our reduced 8-stage decomposition -H-C-CZ-P-H-P-CZ-C- allows execution in the LNN architecture in only $14n{-}4$ two-qubit stages.

\section{Stabilizers and the symplectic group}

We now establish a normal form for stabilizer circuits that eliminates two of the layers of the $11$-layer form given in \cite{ar:ag}, while using same types of layers.  As already mentioned, the stabilizer circuits form a finite group which, modulo the group that is generated by the center and the Pauli subgroup, is isomorphic to the binary symplectic group defined as follows (see also \cite{ar:crss,ar:crss2} and \cite[Chap.~2]{ATLAS}): 

\begin{definition}\label{def:symplectic}
The group $\Sp$ of symplectic matrices of size $2n \times 2n$ with entries over the finite field $\F_2 =\{0,1\}$ is defined as 
$\Sp := \{ A \in \GL : A^t J A = J \}$, where $J = \left[\begin{smallmatrix} \zeromat_n & \onemat_n \\ \onemat_n & \zeromat_n \end{smallmatrix}\right]$, and $\onemat_n$ and $\zeromat_n$ denote the identity matrix and the all zero matrix of size $n\times n$ (the subscript may furthermore be dropped when it is clear what the dimension is; $\zeromat$ may furthermore be used to denote a rectangular matrix), respectively.
\end{definition}

Similarly to \cite{ar:ag} we can work with a tableau representation for symplectic matrices, where we omit column vector $r$ as in \cite{ar:ag}, which corresponds to an overall sign that, if needed, can be obtained via a single layer of $Z$ gates. Definition \ref{def:symplectic} implies that the square block matrix $M=\begin{pmatrix} A & B \\ C & D \end{pmatrix}$ is symplectic if and only if the following four conditions hold: 
\begin{eqnarray}\label{eq:cols}
&& A^tC=C^tA, \quad A^tD+C^tB=\onemat_n, \nonumber \\ 
&& B^tD=D^tB, \quad B^tC+D^tA=\onemat_n.
\end{eqnarray}
In other words, two columns $c_i$ and $c_j$ of $M$ are perpendicular with respect to the symplectic inner product, unless they form one out of $n$ symplectic pairs $(c_i, c_{n+i})$, where $i=0,1,...,n{-}1$, and in which case the symplectic inner product evaluates to $1$. It should be noted that if $M$ is symplectic, so is $M^{-1}$, as the symplectic matrices form a group. As $M^{-1} = \begin{pmatrix} D^t & -B^t \\ -C^t & A^t \end{pmatrix}$, the equation $(M^{-1})^t J M^{-1} = J$ implies that 
the following four conditions hold for $M$ as well: 
\begin{eqnarray}\label{eq:rows}
&& AB^t=BA^t, \quad AD^t+BC^t=\onemat_n, \nonumber \\
&& CD^t=DC^t, \quad CB^t+DA^t=\onemat_n. 
\end{eqnarray}
In other words, also two rows  $r_i$ and $r_j$ of $M$ are perpendicular with respect to the symplectic inner product, unless they form one out of $n$ symplectic pairs $(r_i, r_{n+i})$, where $i=0,1,...,n{-}1$, and in which case the symplectic inner product evaluates to $1$. 

Equations (\ref{eq:cols}) and (\ref{eq:rows}) will be useful later when we bring a given stabilizer circuit, represented as a symplectic matrix, into a suitable normal form. 

The right side action of the stabilizer circuit layers -H-, -P-, and -C- on a symplectic matrix $M$ can be described as follows (see also \cite{ar:ag}): 
\begin{itemize}
\item Right multiplication with a Hadamard gate on qubit $k$ corresponds to exchanging columns $k$ and $n{+}k$ of $M$. 
\item Right multiplication with a Phase gate on qubit $k$ corresponds to the addition modulo 2 of column $k$ of $M$ to column $n{+}k$;
\item Right multiplication with a $\CNOT$ gate with control $j$ and target $k$, $1 \leq j,k \leq n$, corresponds to the addition modulo 2 of column $j$ to column $k$ of $M$ and the addition modulo 2 of column $n{+}k$ to column $n{+}j$ of $M$.
\end{itemize} 
Similarly, the left side action on the rows of $M$ can be defined. 

\begin{figure*}
\includegraphics[width=\textwidth]{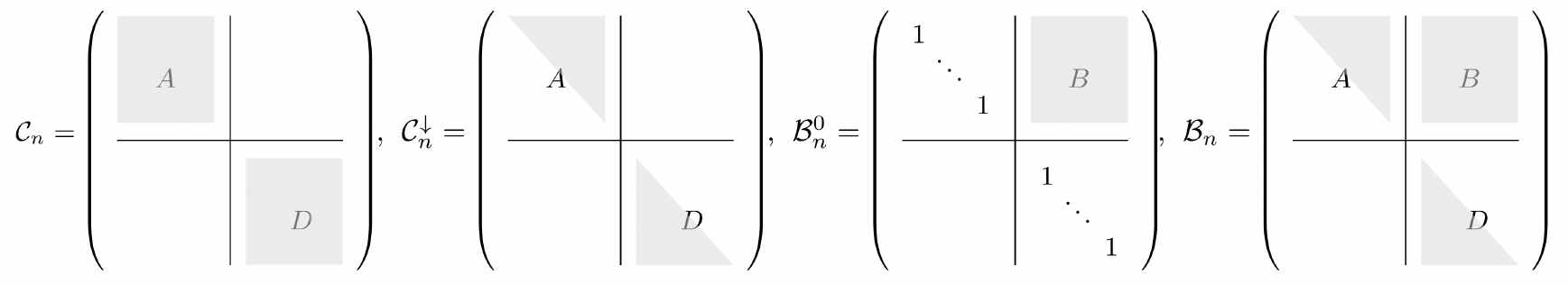}
\caption{\label{fig:Bruhat} Several important subgroups of the stabilizer group on $n$ qubits and a visualization of the structure of the corresponding symplectic matrices: ${\cal C}_n$ corresponds to the group generated by all $\CNOT$ gates. The block $D$ here is equal to $(A^{-1})^t$. If we consider only such $\CNOT(x;y)$ gates where the control $x$ and the target $y$ labels satisfy $x \leq y$, we obtain the subgroup ${\cal C}_n^\downarrow$. In circuit terms, this means that the circuits in ${\cal C}^\downarrow_n$ may be written with only $\CNOT$ gates where the targets are lower than the controls.  The circuits in ${\cal B}^0_n$ correspond to those that can be written in the form ${c}^\downarrow_n p {{c}^\downarrow_n}^{-1}$, where ${c}^\downarrow_n \in {\cal C}^\downarrow_n$ and $p$ is a layer of Phase gates. Necessarily, block $B$ must then be symmetric, {\em i.e.}, $B=B^t$.  Finally, ${\cal B}_n$ denotes a Borel subgroup which is a maximal solvable subgroup of the stabilizer group. Again, matrix block equality $D=(A^{-1})^t$ holds and furthermore $A B^t = B^t A$.  Note that any element $b_n\in{\cal B}_n$ can be written as $b_n = c^\downarrow_n b^0_n$, where $c^\downarrow_n \in {\cal C}^\downarrow_n$ and $b^0_n \in {\cal B}^0_n$.  Note further that ${\cal C}_n$ and ${\cal B}_n$ have $n^2$ Boolean degrees of freedom, whereas ${\cal C}^\downarrow_n$ and ${\cal B}^0_n$ have $n^2/2$ Boolean degrees of freedom.}
\end{figure*}

\section{BN-pairs and Bruhat decomposition}

A property of the symplectic group that we exploit to show an asymptotically optimal decomposition is that this group can be written as a disjoint union 
\begin{equation}\label{eq:bruhat}
\Sp = \bigsqcup_{w \in W} B w B,
\end{equation}
where $B$ is the Borel subgroup of $\Sp$ and $W$ labels a system of representatives of the Weyl group of $\Sp$. For complex Lie group this decomposition is also known as the Bruhat decomposition \cite{Bourbaki:68}. However, even over a finite field such as $\F_2$ the decomposition eq.~(\ref{eq:bruhat}) can be suitably defined using the notion of BN-pairs \cite[Chap.~14]{Aschbacher:2000}, \cite{Tits:74}. As we will see below, we can identify $B$ with a subgroup of $\Sp$ that is isomorphic to a subgroup of the upper triangular matrices, and we can identify $W$ with a wreath product of $\Z_2$ with $S_n$ which corresponds to the group generated by all qubit permutations together with all possible Hadamard gate combinations on $n$ qubits.  

\begin{definition}(BN pair)
Let $G$ be a group and $B, N \subseteq G$ be two subgroups such that $G = \langle B, N \rangle$ and $T := B \cap N$ is a normal subgroup of $N$. 
Let $S$ be a set of generators for $W := N/T$. Denote by $C(w) = BwB$ the double coset corresponding to the representative $w \in W$. If the following two properties hold for all $s \in S$ and all representatives $w \in W$
\begin{itemize}
\item $C(s) C(w) \subseteq C(w) \cup C(sw)$,
\item $s B s^{-1} \not\subseteq B$,
\end{itemize}
then $(B,N)$ is called a {\em BN-pair} and the data $(G, B, N, S)$ is called a {\em Tits system}, see also \cite{Aschbacher:2000,Bourbaki:68,Tits:74}. 
\end{definition}

For the group $G=\Sp$ the subgroup $B$ can be identified with the set ${\cal B}_n$ defined in Fig.~\ref{fig:Bruhat}. Completing the description of the BN-pair in case of $\Sp$ we have to determine the subgroups $N$, $T$, and $W$. 

In case of the finite field $\F_2$ it turns out that $T$ is trivial and $N$ consists of the group generated by all permutation matrices and all Hadamard gates. This means that a set $S$ of generators for $W$ can be defined as 
\begin{eqnarray}\label{eq:gens}
S &:=&\left\{\left[\begin{array}{cccc} 
\zeromat_{k}&&\onemat_k&\\
&\onemat_{n-k} && \zeromat_{n-k}\\
\onemat_{k}&&\zeromat_k&\\
&\zeromat_{n-k} && \onemat_{n-k}
 \end{array}\right] : k =0..n \right\} \nonumber \\
&\bigcup& 
\left\{\left[\begin{array}{cc} 
\tau_i & \\ & \tau_i \end{array} 
\right] : \tau_i = (i, i{+}1), i =1..n{-}1 \right\}.
\end{eqnarray}
The first set in eq.~(\ref{eq:gens} corresponds to the tensor products of Hadamard matrices, namely $W=\{w_k: k=0,1,...,n\}$, where $w_k = H_2^{\otimes k} \otimes \onemat_2^{\otimes n-k}$, whereas the second set corresponds to wire permutations of adjacent wires. Furthermore, we note that 
\begin{eqnarray*}
{\cal B}_n &=& \left\{\left[\begin{array}{cc} A & \zeromat_n \\ \zeromat_n & (A^t)^{-1}\end{array}\right] \left[\begin{array}{cc} \onemat_n & B \\ \zeromat_n & \onemat_n \end{array}\right] \right.\\
&& \left. : A \in \GL, B \in \F_2^{n\times n}, B=B^t \right\},
\end{eqnarray*}
which implies that ${\cal B}_n$ is isomorphic to a subgroup of the upper triangular matrices, {\em i.e.}, in particular, it is a solvable group. This decomposition also implies that there are $n^2/2$ Boolean degrees of freedom in the part corresponding to $\left[\begin{array}{cc} A & \zeromat_n \\ \zeromat_n & (A^t)^{-1}\end{array}\right]$ and $n^2/2$ Boolean degrees of freedom in the part corresponding to $\left[\begin{array}{cc} \onemat_n & B \\ \zeromat_n & \onemat_n \end{array}\right]$ as $B$ is symmetric. Hence, matrices in ${\cal B}_n$ have an overall of $n^2$ Boolean degrees of freedom. 

Finally, note that the elements of the form $\diag(\tau,\tau)$ stabilize the set ${\cal B}^0_n$ as they leave the diagonal part invariant and map the set of symmetric matrices into itself.

\section{Computing the Bruhat decomposition}\label{sec:6}

We first state two lemmas that will be useful later for a step-wise decomposition of a given stabilizer circuit. 

\begin{lemma}[\cite{ar:ag}]\label{lem:symdec} For any symmetric matrix $A\in \F_2^{n \times n}$ there exist
matrices $\Lambda, U\in \F_2^{n\times n}$ such that $A = U U^t + \Lambda$, where $\Lambda$ is diagonal and $U$ invertible and upper triangular.
\end{lemma} 
\begin{proof}
In \cite{ar:ag} a decomposition $M=L L^t+\Lambda$ was derived, where $L$ is lower triangular. By conjugating this expression with a permutation matrix that exchanges the rows $(1,n), (2,n-1), ...$ we see that the same proof also gives rise to a decomposition into $M=U U^t + \Lambda'$ with $U$ upper triangular and some diagonal matrix $\Lambda'$. 
\end{proof}
\begin{corollary} Any matrix in ${\cal B}_n$ can be written in the form -C-P-C-P- or alternatively in the form -P-C-P-C- with all -C- layers consisting of gates in ${\cal C}_n^\downarrow$. 
\end{corollary}
\begin{proof}
To see this, we first apply Lemma \ref{lem:symdec} decompose a given matrix $A = \left[\begin{array}{cc} \onemat_n & B \\ \zeromat_n & \onemat_n \end{array}\right]$ into the product 
$A = \left[\begin{array}{cc} \onemat_n & U U^t \\ \zeromat_n & \onemat_n \end{array}\right]\left[\begin{array}{cc} \onemat_n & \Lambda \\ \zeromat_n & \onemat_n \end{array}\right]$. Now, the first factor can be implemented in the form -C-P-C- and we get the following overall circuit of the form -C-P-C-P- for $A$:
\[
A = 
\left[\begin{array}{cc} U & \zeromat \\ \zeromat & (U^t)^{-1} \end{array}\right]
\left[\begin{array}{cc} \onemat & \onemat \\ \zeromat & \onemat \end{array}\right]
\left[\begin{array}{cc} U^{-1} & \zeromat \\ \zeromat & U^t \end{array}\right]
\left[\begin{array}{cc} \onemat & \Lambda \\ \zeromat & \onemat \end{array}\right].
\]
Clearly, the -C- layers are in ${\cal C}_n^\downarrow$. The other decomposition -P-C-P-C- is obtained similarly, by factoring out the $\Lambda$ component on the left. 
\end{proof}

\begin{lemma}\label{lem:symplecticLPU} For any matrix $M \in \F_2^{n\times 2n}$ that is the lower $n \times 2n$ part of a $2n \times 2n$ symplectic matrix, there exist a lower triangular matrix $L$, an upper triangular matrix $U$, permutation matrices $\sigma,\tau \in S_n$, and $k$, $0 \leq k \leq n$, such that 
\[
M{=}
L \, \sigma \, \left[
\begin{array}{cc|cc}
\onemat_k & \zeromat & D_1 & D_2 \\ 
\zeromat & \zeromat & \zeromat & \onemat_{n-k} 
\end{array}
\right] \diag(\tau, \tau) \, \diag(U, (U^{-1})^t).
\]
\end{lemma} 
\begin{proof}
The main idea is to use the fact that any matrix $M \in \F_2^{n \times 2n}$ can be decomposed into the product of a triangular matrix, a permutation pattern ({\em i.e.}, a matrix that has at most one non-zero entry in each row and column), and another triangular matrix. LU decomposition with pivoting is a special case of this decomposition \cite[Theorem 3.4.2]{GvL:2000}, \cite[Theorem 3.5.7]{HJ:1985}, however, in our situation we cannot assume that we know the pivoting of the matrix. Using $L$, $P$, and $U$ as shorthand for lower triangular, permutation pattern, and upper triangular matrices, it is known that all four combinations $M=LPL=LPU=UPL=UPU$ are possible, see, {\em e.g.}, \cite{Strang:2012} for a discussion. For instance, for $LPL$ we start in the upper right hand corner of $M$ and eliminate the non-zero matrix entries going down and left.  For $LPU$, we start in the upper left corner and eliminate the non-zero matrix entries going down and right.  The remaining pattern defines the $P$-part of the matrix. 

Since, by assumption, $M$ is a part of the $2n \times 2n$ symplectic matrix, we obtain that ${\rm rk}(M)=n$ which means that using an $LPU$ decomposition on the left $n{\times}n$ block of $M$ we can find $L_1$ and $U_1$ such that 
$L_1 M \diag(U_1, (U_1^{-1})^t) = 
\left[P_1 | M_1\right]$, where $P_1$ is a permutation pattern and $M_1$ is another matrix. By considering the support of $P_1$ we can define row indices $R := \{i \in \{1,2,..., n\}: (P_1)_{i,*} = 0^n \}$ and column indices $C := \{ j \in \{1,2,..., n\}: (P_1)_{*,j} = 0^n \}$. If $k := {\rm rk}(P_1)$, then clearly $|R|=|C|=n{-}k$. Using an $LPL$ decomposition on the restriction of the right block $M_1$ of this new matrix to the rows and columns in $R \times C$, we therefore obtain 
$L_2$, $L_2^\prime$, and permutation matrices $\sigma$, $\tau$ such that 
$\sigma L_2 L_1 M \diag(U_1, (U_1^{-1})^t) \diag(((L_2^\prime)^{-1})^t,L_2^\prime)\diag(\tau,\tau)= 
\left[
\begin{array}{cc|cc}
\onemat_k & \zeromat & D_1 & D_2 \\ 
\zeromat & \zeromat & \zeromat & \onemat_{n-k} 
\end{array}
\right]$ for certain $D_1 \in \F_2^{k \times k}$ and $D_2 \in \F_2^{k \times (n-k)}$.
\end{proof}

\begin{theorem}\label{thm:cpcphc}
Any Clifford circuit on $n$ qubits can written in the form -P-C-P-C-H-C-P-C-P-. 
\end{theorem}
\begin{proof}
We start with the $2n\times 2n$ symplectic matrix $M$ of the form 
\[ 
M = \left[ \begin{array}{c|c}
A & B  \\ \hline \\[-2ex]
C & D \end{array}\right],
\]
where $A$, $B$, $C$, and $D$ are in $\F_2^{n\times n}$. We next give an algorithm that synthesizes $M$ in a canonical form.  The algorithm proceeds in several steps, by clearing out the entries of $M$ via left-hand and right-hand multiplications by other matrices, until finally only a permutation matrix remains, which then corresponds to a Hadamard layer up to a permutation of qubits. 

Step 1.  We apply Lemma \ref{lem:symplecticLPU} to the submatrix $\left[ C | D \right]$.  Note that since $M \in \Sp$ we have $C^t D = D^t C$, {\em i.e.}, the conditions to the lemma are satisfied and we can find a lower triangular matrix $L \in \GL$ and an upper triangular matrix $U \in \GL$ and two permutation matrices $\sigma, \tau \in S_n$ such that 
\begin{align}
& \sigma L 
[ C | D ] 
\left[
\begin{array}{cc}
U & \zeromat \\ \zeromat & (U^t)^{-1}
\end{array}
\right]
\left[
\begin{array}{cc}
\tau & \zeromat \\ \zeromat & \tau
\end{array}
\right] \nonumber\\
& = 
\left[
\begin{array}{cc|cc}
\onemat_k & \zeromat & D_1 & D_2 \\ 
\zeromat & \zeromat & \zeromat & \onemat_{n-k} 
\end{array}
\right], \nonumber
\end{align}
where $0 \leq k \leq n$, $D_1 \in \F_2^{k \times k}$, $D_2 \in \F_2^{k \times (n-k)}$, and $\zeromat$ denotes all zero-matrices of the appropriate sizes.  Application of these operations to the initial matrix forces some simplifications: 
\begin{align}
& \phantom{=} \left[
\begin{array}{c|c}
\sigma & \zeromat \\ \zeromat & \sigma 
\end{array}
\right]
\left[
\begin{array}{c|c}
(L^t)^{-1} & \zeromat \\ \zeromat & L
\end{array}
\right]
M 
\left[
\begin{array}{c|c}
U & \zeromat \\ \zeromat & (U^t)^{-1} 
\end{array}
\right]
\left[
\begin{array}{c|c}
\tau & \zeromat \\ \zeromat & \tau
\end{array}
\right] \nonumber \\
& = 
\left[
\begin{array}{cc|cc}
 A_1 & A_2 & B_1 & B_2  \\
 A_3 & A_4 & B_3 & B_4  \\
\onemat_k & \zeromat & D_1 & D_2 \\ 
\zeromat & \zeromat & \zeromat & \onemat_{n-k} 
\end{array}
\right] =: M_1. \nonumber
\end{align}
Here, $A_2=\zeromat$ (implying $A_4=\onemat_{n-k}$) and $A_1$ is symmetric because of the symplectic condition between the last two block rows and the first two block rows of this matrix.

Step 2. We left multiply the matrix $M_1$ by a matrix in ${\cal B}_n^0$, as follows: 
\begin{align}
& \phantom{=} \hspace{-2mm}\left[
\begin{array}{cccc}
\onemat_k & \zeromat    & A_1 & A_3^t \\ 
\zeromat & \onemat_{n-k} & A_3 & \zeromat \\ 
\zeromat & \zeromat & \onemat_k & \zeromat \\
\zeromat & \zeromat & \zeromat & \onemat_{n-k}
\end{array}
\right] \hspace{-1mm}
\left[
\begin{array}{cccc}
 A_1 & \zeromat & B_1 & B_2  \\
 A_3 & \onemat_{n-k} & B_3 & B_4  \\
\onemat_k & \zeromat & D_1 & D_2 \\ 
\zeromat & \zeromat & \zeromat & \onemat_{n-k} 
\end{array}
\right] \nonumber \\
& =
\left[
\begin{array}{cccc}
 \zeromat & \zeromat & B_1^\prime & B_2^\prime  \\
 \zeromat & \onemat_{n-k} & B_3^\prime & B_4^\prime  \\
\onemat_k & \zeromat & D_1 & D_2 \\ 
\zeromat & \zeromat & \zeromat & \onemat_{n-k} 
\end{array}
\right] =: M_2. \nonumber
\end{align}
Note that since $A_1$ is symmetric the matrix $\left[\begin{smallmatrix} A_1 & A_3^t \\ A_3 & \zeromat \end{smallmatrix}\right]$ is symmetric as well.  We can apply Lemma \ref{lem:symdec} to obtain a decomposition of this upper triangular symplectic matrix applied from the left as -P-C-P-C-, where all -C- layers are in ${\cal C}_n^\downarrow$. 

Step 3. Note that because of the symplectic condition between columns one and three of $M_2$ we must have $B_1^\prime = \onemat_k$ and $D_1$ is symmetric.  Similarly, the symplectic condition between columns two and four of $M_2$ implies that $B_2^\prime = \zeromat$ and $B_4^{\prime}$ is symmetric.  Moreover, by considering the symplectic condition between rows two and three, which needs to be zero, we obtain that $B_3^\prime = D_2^t$. We can therefore apply a final column operation to $M_2$ to clear out the remaining entries by multiplying on the right

\begin{align}
& \phantom{=}  
\left[
\begin{array}{cccc}
 \zeromat & \zeromat & \onemat_k & \zeromat  \\
 \zeromat & \onemat_{n-k} & D_2^t & B_4^\prime  \\
\onemat_k & \zeromat & D_1 & D_2 \\ 
\zeromat & \zeromat & \zeromat & \onemat_{n-k} 
\end{array}
\right] 
\left[
\begin{array}{cccc}
\onemat_k & \zeromat    & D_1 & D_2 \\ 
\zeromat & \onemat_{n-k} & D_2^t & B_4^\prime \\ 
\zeromat & \zeromat & \onemat_k & \zeromat \\ 
\zeromat & \zeromat & \zeromat & \onemat_{n-k} 
\end{array}
\right] \nonumber \\
& = 
\left[
\begin{array}{cccc}
 \zeromat & \zeromat & \onemat_k & \zeromat  \\
 \zeromat & \onemat_{n-k} & \zeromat & \zeromat  \\
\onemat_k & \zeromat & \zeromat & \zeromat \\ 
\zeromat & \zeromat & \zeromat & \onemat_{n-k} 
\end{array}
\right] =: M_3. \nonumber
\end{align}
As in Step 2, the symmetric (this follows from its block expression and the previously established notion that both $D_1$ and $B_4^{\prime}$ are symmetric) matrix $\left[\begin{smallmatrix} D_1 & D_2 \\ D_2^t & B_4^\prime \end{smallmatrix}\right]$ can be decomposed using Lemma \ref{lem:symdec} to obtain a representation of the overall upper triangular matrix applied from the right in the form -P-C-P-C-, where again all -C- layers are in ${\cal C}_n^\downarrow$. 

The final matrix $M_3$ corresponds to a sequence of Hadamard gates applied to the first $k$ qubits. 

Overall, we applied the sequence 
\[
U_1 \pi_1 T_1 M T_2 \pi_2 U_2 = H,
\]
where $H$ is a product of Hadamard matrices applied to the first $k$ basis states, $T_1, T_2\in {\cal C}_n^\downarrow$, $\pi_1,\pi_2 \in S_n$, and $U_1, U_2 \in {\cal B}_n^0$. Multiplication by inverses from both sides yields 
\[
M = T_1^{-1} \pi_1^{-1} U_1^{-1} H U_2^{-1} \pi_2^{-1} T_2^{-1}.
\]
Now, notice that permutations stabilize ${\cal B}_n^0$, {\em i.e.}, we can find $V_1, V_2 \in {\cal B}_n^0$ such that 
\[
M = T_1^{-1} V_1 \pi_1^{-1} H \pi_2^{-1} V_2 T_2^{-1}.
\]
Note that $V_1$ is of the form -C-P-C-P- with the first -C- layer in ${\cal C}^\downarrow$, {\em i.e.}, $T_1^{-1}$ and $V_1$ can be combined into one matrix $W_1 \in {\cal B}_n$.  Similarly, $V_2$ can be written in the form -P-C-P-C- and therefore $V_2$ and $T_2^{-1}$ can be combined into one matrix $W_2 \in {\cal B}_n$.  Note finally that we can implement $\pi_1^{-1} H \pi_2^{-1}$ using a single layer of Hadamard gates $H_1$ acting non-trivially on some $k$ qubits, and merge the qubit swapping stage with either $W_1$ or $W_2$. Overall, we have that $M$ can be written as
\[
M = W_1 H \pi W_2 \in \text{-C-P-C-P-H-P-C-P-C-.}
\]
Since -C-P-C-P- circuit can be written as a -P-C-P-C- circuit, the claimed decomposition follows. 
\end{proof}

Combining the results of Theorems \ref{thm:1} and \ref{thm:cpcphc}, and Corollary \ref{col:2} allows to obtain the main result of this paper,

\begin{corollary}\label{cor:main}
An arbitrary stabilizer circuit can be written as a 7-stage layered decomposition -C-CZ-P-H-P-CZ-C-.  It is executable in the LNN architecture as a two-qubit gate depth-$(14n-4)$ circuit. 
\end{corollary}

\begin{corollary}
Defining $B := {\cal B}_n$ and $N (=W)$ to be the group generated by -H- and all wire permutations we obtain that $B$ and $N$ define a $BN$-pair for $\Sp$. 
\end{corollary}
\begin{proof}
From Theorem \ref{thm:cpcphc} we obtain, in particular, that $B$ and $N$ generate the entire group $\Sp$. Clearly we have that $T = B \cap N$ is trivial, {\em i.e.}, it is normal in $N$.  The stated property $s B s^{-1} \not\subseteq B$ for all $s \in S$ clearly holds for our choice of the generator set $S$ in eq.~(\ref{eq:gens}), as Hadamard as well as qubit swaps do not preserve the directional $\CNOT$ gates. Finally, to establish the coset multiplication rule $C(s)C(w) \subseteq C(w) \cup C(sw)$ we use \cite[Chap.~V.6]{Brown:89}. 
\end{proof} 

\begin{corollary}
The Bruhat decomposition gives rise to an asymptotically tight parametrization of all $2^{2n^2 + O(n)}$ stabilizer circuits.
\end{corollary}
\begin{proof}
This is a direct consequence of the decomposition into layers of the form -C-P-C-P-H-P-C-P-C- proved in Theorem \ref{thm:cpcphc}. From the proof of the theorem we see that the -C-P-C-P- and the -P-C-P-C- layers correspond to the elements of ${\cal B}_n$, each of which has $n^2 + o(n^2)$ Boolean degrees of freedom.  This yields the claimed statement. 
\end{proof}

\section{Normal form for stabilizer circuits}

The Bruhat decomposition eq.~(\ref{eq:bruhat}) allows us to characterize the possible block structures that stabilizer operators might have when considered as a unitary matrix of size $2^n \times 2^n$ and how they behave under multiplication.  

\begin{definition} Let $C$ be a stabilizer circuit. Let $B \cdot w(C) \cdot B$ denote the unique double coset that $C$ lies in. Then we can represent $w(C)$ by an element in $\Z_2^n \rtimes S_n$, or equivalently by a matrix of the form $U \pi$, where $U$ is a tensor product of $k$ Hadamard matrices, where $1 \leq k \leq n$ and $\pi$ is a permutation matrix of $n$ wires. By rearranging the non-identity Hadamard operators, we can represent such an element $U\pi$ in a form $\sigma (\onemat_2^{\otimes {n-k}} \otimes H^{\otimes k}) \tau$, where $\pi=\sigma\tau$. We call $(k, \sigma, \tau)$ the block structure of $C$. 
\end{definition} 

Note that whereas $U$ and $\pi$ are unique, in general $\sigma$ and $\tau$ are not, as there is a degree of freedom corresponding to elements in $S_{n-k} \times S_k$. However, the collection $I \subset \mathbb{Z}_{2^n}^2$ defined as $I := \left\{ (i,j): |C_{i,j}| = \frac{1}{\sqrt{2}^k} \right\}$ is uniquely defined by $C$ and the corresponding block structure $(k,\sigma,\tau)$. As a corollary to Theorem \ref{thm:cpcphc} we obtain the following multiplication rule for block structures. 

\begin{corollary}
Let $C_1$ and $C_2$ be stabilizer circuits with block structures $(k_1,\sigma_1,\tau_1)$ and $(k_2, \sigma_2, \tau_2)$, respectively. Then the block structure of $C_1 C_2$ is of the form $(m,\sigma_3,\sigma_3^{-1} \sigma_1 \tau_1 \sigma_2 \tau_2)$, where $0 \leq m \leq k_1{+}k_2$ and $\sigma_3 \in S_n$. 
\end{corollary}
\begin{proof}
Let $w(C_1) = U \pi$ denote the representative of $C_1$ in the Weyl group. Write $w(C_1)$ as a product over the generators $S = S_h \cup S_p$ where $S_h = \left\{h_i = \onemat_i: i \in \{1,2,...,n\}\right\}$ and $S_p = \left\{ p_i = (i,i+1) : i \in \{1,2,...,n{-}1\} \right\}$. As the Weyl group is a semidirect product we can collect the factors corresponding to $S_h$ and $S_p$ together and write $w(C_1) = \prod_{i=1}^{k_1} h_{t_{1,i}} \prod_{j=1}^{n_1} p_{s_{1,i}}$. We similarly write $w(C_2)$  and note that as $W$ is a semidirect product, we get $w(C_1)w(C_2)= 
\prod_{i=1}^{k_1} h_{t_{1,i}} \prod_{i=1}^{k_2}  h_{\pi t_{2,i}} \prod_{j=1}^{n_1} p_{s_{1,i}} \prod_{j=1}^{n_2} p_{s_{2,i}}$. As there might be cancellation between the Hadamard matrices $\prod_{i=1}^{k_1} h_{t_{1,i}}$ and $\prod_{i=1}^{k_2}  h_{\pi t_{2,i}}$, we obtain that the Hadamard block can have $m$ non-trivial factors where $0 \leq m \leq k_1{+}k_2$. The permutational parts multiply, {\em i.e.}, we conclude that the permuted block structure of the product is of the claimed form. 
\end{proof}

\section{Conclusion}

In this paper, we reduced the 11-stage computation -H-C-P-C-P-H-P-C-P-C- \cite{ar:ag} into the 9-stage decomposition -C-P-C-P-H-P-C-P-C- relying on the Bruhat decomposition of the symplectic group.  We showed that all -C- stages in our 9-stage decomposition correspond to upper triangular matrices.  This leads to an asymptotically tight parameterization of the stabilizer group, matching its number of $2^{2n^2+O(n)}$ degrees of freedom.  We then derived a 7-stage decomposition of the form -C-CZ-P-H-P-CZ-C-, that relies on the stage -CZ-, not considered by \cite{ar:ag}.  We showed evidence that the -CZ- stage is likely superior to the comparable -C- stage.  Indeed, the number of the Boolean degrees of freedom in the -CZ- stage is only about a half of that in the -C- stage, two-qubit gate counts for optimal implementations of -CZ- circuits remain smaller than those for -C- circuits (see Table \ref{tab:main}), and -CZ- computations were possible to implement in a factor of $2.5$ less depth than that for -C- stage computations over LNN architecture.  

We reported a two-qubit gate depth-$(14n-4)$ implementation of stabilizer unitaries over the gate library $\{\H,\P,\CNOT\}$, executable in the LNN architecture.  This improves previous result, a depth-$25n$ circuit \cite{ar:ag,ar:kms} executable over LNN architecture.

Our 7-stage construction can be written in $16$ different ways, by observing that -C-CZ-P- can be written in $4$ different ways: -C-CZ-P-, -C-P-CZ-, -P-CZ-C-, and -CZ-P-C-. For the purpose of practical implementation we believe a holistic approach to the implementation of the 3-layer stage -P-CZ-C- may be due, where the linear reversible function $g(x_1,x_2,...,x_n)$ is implemented by the $\CNOT$ gates such that the intermediate linear Boolean functions generated go through the set that allows implementation of the phase polynomial part.

\section{Acknowledgements}
DM thanks Yunseong Nam from IonQ for helpful discussions.  MR thanks Jeongwan Haah and Vadym Kliuchnikov from Microsoft Research for discussions.  Circuit diagrams were drawn using qcircuit.tex package, \href{http://physics.unm.edu/CQuIC/Qcircuit/}{http://physics.unm.edu/CQuIC/Qcircuit/}.

\end{document}